\documentclass[journal]{IEEEtran}

\usepackage{amsmath}
\usepackage{bm}
\usepackage{graphicx}
\usepackage[ruled,linesnumbered]{algorithm2e}
\usepackage{algpseudocode}
\usepackage{amsfonts,amssymb}
\usepackage{mathrsfs}
\usepackage{multirow}
\usepackage{fmtcount}
\usepackage{url}
\usepackage{setspace}
\usepackage{blindtext}
\usepackage{xcolor,colortbl}
\usepackage{fancyhdr}
\usepackage{booktabs} 
\usepackage{dblfloatfix}
\usepackage{anyfontsize}
\usepackage{enumerate}
\usepackage{ntheorem}
\usepackage{cases}
\usepackage{caption}

\theoremheaderfont{\bfseries\itshape}
\theorembodyfont{\upshape}
\theoremseparator{:}

\newtheorem{lemma}{Lemma}
\newtheorem{definition}{Definition}
\newtheorem*{proof}{Proof}
\usepackage{diagbox}

\usepackage{soul}
\DeclareRobustCommand\comment[1]{}
\DeclareRobustCommand\jiaqi[1]{}

\begin{document}
\title{\vspace{-1.0em} Adaptive Task Offloading for Space Missions:\\ A State-Graph-Based Approach}

\author{Jiaqi Cao,\ 
		Shengli Zhang,~\IEEEmembership{Senior Member,~IEEE,}
		Mingzhe Wang,\ 
		Qingxia Chen,\ 
		Houtian Wang,\ 
		Naijin Liu\ 
		\vspace{-2.0em}
	
	\thanks{Naijin Liu is the corresponding author.}      
	\thanks{Jiaqi Cao and Shengli Zhang are with Shenzhen University, Shenzhen, 518052, P.R. China (e-mail: jiaqicao@szu.edu.cn; zsl@szu.edu.cn).}
	\thanks{Mingzhe Wang is with Tsinghua University (e-mail: wmzhere@gmail.com).}
	\thanks{Qingxia Chen, Houtian Wang and Naijin Liu are with Qian Xuesen Laboratory of Space Technology, China Academy of Space Technology (e-mail: chenqingxia@qxslab.cn; wanghoutian@qxslab.cn; liunaijin@qxslab.cn).}
}

\maketitle
\thispagestyle{fancy}

\begin{abstract}
	Advances in space exploration have led to an explosion of tasks. 
	Conventionally, these tasks are offloaded to ground servers for enhanced computing capability, or to adjacent low-earth-orbit satellites for reduced transmission delay. However, the overall delay is determined by \emph{both} computation and transmission costs. The existing offloading schemes, while being highly-optimized for \emph{either} costs, can be abysmal for the overall performance. The computation-transmission cost dilemma is yet to be solved.
	
	In this paper, we propose an adaptive offloading scheme to reduce the overall delay. The core idea is to jointly model and optimize the transmission-computation process over the entire network. Specifically, to represent the computation state migrations, we generalize graph nodes with multiple states. In this way, the joint optimization problem is transformed into a shortest path problem over the state graph. We further provide an extended Dijkstra's algorithm for efficient path finding. Simulation results show that the proposed scheme outperforms the ground and one-hop offloading schemes by up to 37.56\% and 39.35\% respectively on SpaceCube v2.0.
\end{abstract}

\IEEEpeerreviewmaketitle%
	
\section{Introduction}\label{Introduction}
\IEEEPARstart{T}{ask} offloading is an important issue for space applications such as remote sensing (RS)~\cite{plaza2007high}. 
Advances in space technology have accelerated the exploration and exploitation of space resources. Meanwhile, massive tasks with different data volumes and computational requirements make delay reduction more challenging.
%
Conventionally, these tasks are routed via spacecraft, such as low-earth-orbit (LEO) satellites, to ground servers near the destination for computation~\cite{9651919,7805169}. 
With the development of onboard computing, these tasks are also offloaded to LEO satellites within one hop for computation, the results of which are then routed to the destination~\cite{9148261}.

However, these conventional offloading schemes are susceptible to a \emph{transmission-computation cost dilemma}, because their extreme offloading targets prevent them from coping with task diversity. 
Specifically, the \emph{ground offloading scheme} reduces computation delay by the powerful computing capability of ground servers. However, for tasks with large data volumes, the transmission delay can be high, as large amounts of raw data need to be transmitted to distant destinations. 
The \emph{one-hop offloading scheme} reduces transmission delay by performing onboard computing. However, for tasks with high computational requirements, the computation delay can be high because computing resources within one-hop range are limited. 
Briefly, conventional offloading schemes are devoted to reducing transmission or computation delay, but may lead to excessive delay for the other; thus creating a transmission-computation cost dilemma.
An intuitive solution to escape this dilemma is extending the offloading targets. 
LEO satellites located between the source and destination could be a promising solution.
This is because offloading tasks to these satellites not only achieves lower transmission delay than the ground offloading scheme by performing onboard computing, but also provides more computing resources than the one-hop offloading scheme. 
Thus, LEO satellites beyond one-hop range hold the promise of avoiding both high transmission delays and computation delays existing in conventional schemes.


Novel targets bring novel challenges. While offloading tasks to LEO satellites beyond one hop may capture more computing resources to reduce the computation delay, multi-hop transmission of raw data increases the transmission delay. It is difficult to balance the transmission and computation costs to minimize the overall delay. To address this challenge, the transmission and computation processes should be optimized jointly.

In this paper, we propose a state-graph-based task offloading scheme that adaptively selects the optimal offloading target according to the characteristics of tasks. 
By integrating all potential offloading targets, and representing both computation and transmission processes in the state graph, the proposed scheme jointly optimizes computation and transmission delays. In this way, the optimal offloading target and path with the minimal overall delay of any task can be obtained. The main contributions of this paper are summarized as follows.
  
%
%

\begin{itemize}
	\item We propose a state graph as a mathematical optimization tool. It can represent migrations between different nodes and between different states of the same node. Additionally, we propose a low-complexity extended Dijkstra's algorithm that finds the shortest path in the state graph. 
	\item We enrich the offloading targets with LEO satellites beyond one-hop range to resolve the transmission-computation cost dilemma. To this end, we construct the network as a state graph with two states ( uncomputed v.s. computed) and transform the task offloading problem into a shortest path problem in the state graph.
	\item We propose an adaptive task offloading scheme which applies the extended Dijkstra's algorithm to the constructed state graph. On SpaceCube v2.0, it outperforms the ground offloading scheme and the one-hop offloading scheme by up to 37.56\% and 39.35\%, respectively.  
\end{itemize}

In rest of this paper, the system model is introduced in Section~\ref{SystemModel}. Section~\ref{StateGraph} propose a state graph and an extended Dijkstra's algorithm. 
An adaptive task offloading scheme is proposed in Section~\ref{AdaptiveScheduling}. Performance evaluations are given in Section~\ref{evaluation}. Conclusions are drawn in Section~\ref{Conclusions}.

\section{System Model}\label{SystemModel}
Due to the superiority in latency, cost, development cycle, etc., LEO satellite networks are deemed as the most prospective satellite communication system~\cite{agasid2015small}. 
Existing onboard computing systems can provide up to thousands of Giga floating-point operations per second (GFLOPS) of computing capability~\cite{MOOG}, making them a promising solution for space computation offloading.

In this paper, we investigate the space task offloading problem in the LEO-satellite-based multi-tier network shown in Fig.~\ref{NetworkModel}. Without loss of generality, we take a Walker Star constellation~\cite{walker1984satellite} with polar orbits as an example of LEO satellite networks. 
All circular orbits are evenly distributed over $ 180^{\circ} $ range, traveling north on one side of the Earth, and south on the other side. They cross each other only over the North and South poles. LEO satellites are uniformly distributed over each orbit. 
Each LEO satellite has four inter-satellite links (ISLs) with its neighbors where two are intra-plane and two are inter-plane. The ISLs in cross-seam and Polar Regions are switched off due to high dynamic motions.  
\begin{figure}[h]
	\centering
	\vspace{-0.5em}
	\setlength{\abovecaptionskip}{-0.cm}
	\setlength{\belowcaptionskip}{-0.cm}
	\includegraphics[width=0.48\textwidth]{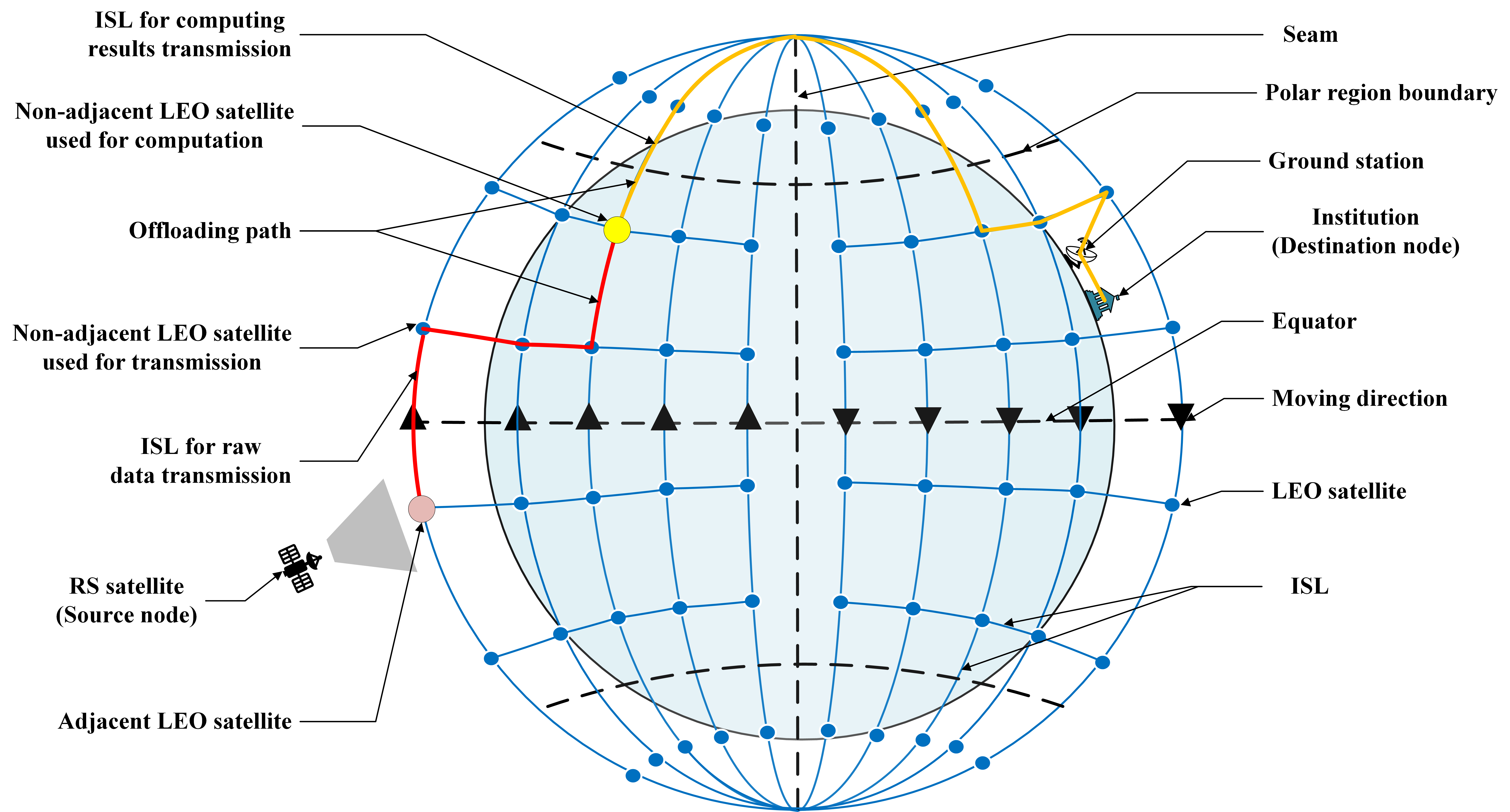}\\ 
	\begin{center}
		\caption{LEO-satellite-based space task offloading network. (Tasks are generated by spacecraft such as RS satellites. Institutions on the ground, such as the meteorological center, are the destinations of these tasks. The offloading path above is an example of offloading tasks to a LEO satellite beyond one-hop range for computation.)}\label{NetworkModel}
	\end{center}
	\vspace{-1.8em}	 
\end{figure}

The offloading path in Fig.~\ref{NetworkModel} is an example of offloading tasks to a LEO satellite beyond one-hop range. On this path, tasks are offloaded to the yellow node for computation which is a LEO satellite located beyond one-hop range from the source. Other LEO satellites on this path are used for data transmission. Accordingly, the red part of this path indicates the raw data transmission; whereas, the yellow part represents the computing results transmission. In the following sections, we dedicate to find the offloading path that achieves the lowest overall delay in the network shown in Fig.~\ref{NetworkModel}.

\section{Path Finding in State Graph}\label{StateGraph}

To solve the joint optimization problem presented in Section~\ref{Introduction}, the cost representation of the transmission and computation processes should first be determined.
%
Conventionally, the cost of the transmission process is represented as the edge weight between two nodes. 
However, for the computation process, both the uncomputed and computed states take place at the same node; thus, the computational states need to be differentiated on each node. 
Therefore, we proposed a state graph whose nodes have multiple states. The computation process can be represented as the computational state migration of nodes. 
We also proposed an extended Dijkstra's algorithm to find the shortest path in the state graph. With these mathematical tools, the transmission and computation processes can be expressed and optimized jointly. 
\vspace{-0.5em}
\subsection{State Graph}
%
\begin{definition}[State graph]
	Assuming that $ G = (V, E) $ is a graph, where $ V $ is the set of nodes and $ E $ is a function from $ V \times V $ to $ \mathbb{R} $. 
	Then, a state graph is an ordered triple $ G_{S} = (\overline{V}, \overline{E}, T ) $ comprising:\\\vspace{-0.8em}
	\begin{itemize}
		\item $ \overline{V} = \{ V_{k,s}\ |\ k \in K,\ s\in V \} $, a set of nodes related to a set of ordered states, where $ V_{k,s} $ represents the node $ s $ in state $ k $. $ K $ is the set of state.
		\item $ \overline{E} = \{ E_{k}\ |\ k\in K\}$, a set of functions related to multiple states, where $  E_{k}(s,s') ( s,s'\in V)$ represents the function from node pair (i.e., edge) $ (V_{k,s},V_{k,s'}) $ in state $ k $ to $ \mathbb{R} $ (i.e., $ E_{k}:\ V \times V \to \mathbb{R} $).
		\item $ T = \{ T_{k}\ |\ k\in K \} $, a set of functions related to multiple states, where $ T_k(s)\ (s \in V) $ represents the function from node pair $ (V_{k,s},V_{k+1,s}) $ to $ \mathbb{R} $ (i.e., $ T_{k}:\ V \to \mathbb{R} $).
	\end{itemize}
\end{definition}

Each node in $ G_{S} $ has $ |K| $ states. Thus, there are $ |K| $ functions for calculating the edge weights between nodes in the same state and $ |K|-1 $ functions for calculating the edge weights between adjacent states of the same node.

\vspace{-0.5em}
\subsection{Finding the Shortest Path}\label{ShortestPathsInStatusGraph}
Since state graphs have some differences compared with simple graphs, we first give a definition of paths in state graphs, and then propose an extended Dijkstra's algorithm which finds the shortest path in state graphs. 
\begin{definition}[Path]
	A path $ P_{S} $ in the state graph $ G_{S} = (\overline{V}, \overline{E}, T ) $ is a finite sequence of distinct edges $\large(e_1, e_2, \dots, e_l)\ \big(e_l=(V_{k_l,s_l},V_{k_{l+1},s_{l+1}}),\ l\in \mathbb{Z}^+ \big) $ for which there is a sequence of nodes $ (V_{k_1,s_1},V_{k_2,s_2},\dots,V_{k_{l+1},s_{l+1}}) $ subjecting to the following constraints:
	\begin{itemize}
		\item $ \forall \ l\in \mathbb{Z}^+ (k_{l-1} \leq k_l )$;
		\item $ \forall \ i,j\in\{1,2,\dots,l+1\} (i\neq j \land k_i=k_j \rightarrow s_i\neq s_j )$.
	\end{itemize}
\end{definition}

In the above definition, $ V_{k_1,s_1} $ and $ V_{k_{l+1},s_{l+1}} $ are the source and destination nodes of path $ P_{S}(V_{k_1,s_1},V_{k_{l+1},s_{l+1}}) $. $V_{k_l,s_l} $ and $ V_{k_{l+1},s_{l+1}} $ are the head and tail of edge $ e_{l}=(V_{k_l,s_l},V_{k_{l+1},s_{l+1}}) $. 
The first constraint indicates that for any edge in $ P_{S} $, the state index of its head is not greater than that of its tail. 
The second constraint indicates that any two nodes of $ P_{S} $ must be different from each other. In other words, two nodes of the same path can be two different nodes belonging to $ V $ in the same state or can be the same node belonging to $ V $ in different states.

Accordingly, the path length of $ P_{S}(V_{k_1,s_1},V_{k_{l+1},s_{l+1}}) $ can be calculated as $ L(P_{S})=\sum_{e_l\in P_{S}}{\omega_l} $, where $ \omega_l $ is the weight of $ e_l $ and can be calculated as:
\begin{small}
	\begin{equation}
		\setlength{\abovedisplayskip}{3pt}
		\setlength{\belowdisplayskip}{-3pt}
		{\omega_l}=
		\left\{
		\begin{array}{ll}
			E_{k_l}(s_l,s_{l+1}),\ k_{l}=k_{l+1}, \\
			T_{k_l}(s_l),\  k_{l} \neq k_{l+1}.
		\end{array}\label{Equ:1a}
		\right.
	\end{equation}	
\end{small}


\setlength{\textfloatsep}{-6pt}
\begin{algorithm}[t]\scriptsize
	\caption{The Extended Dijkstra's Algorithm}\label{ModifiedDijkstra}
	\KwData{The state graph $ G_{S} = (\overline{V}, \overline{E}, T ) $ and the state set $ K $.}
	\KwData{The source node $ u $ and the destination node $ v $.}
	
	\KwResult{The shortest path $ \pi $ from $ V_{1,u} $ to $ V_{|K|,v} $ and its length $ L_{\pi} $.}
	
	Initialize $ Q \leftarrow \emptyset $ and $ \pi \leftarrow \emptyset $;
	
	\ForEach{$ V_{k,s} \in\overline{V} $}{
		$D[k,s] \leftarrow \infty$;			
		
		$ P[k,s] \leftarrow $ undefined;
		
		$ Q \leftarrow Q\cup \{V_{k,s}\} $;
	}
	
	$ D[1,u] \leftarrow 0 $;
	
	\While{$ Q \neq \emptyset $}{
		$ \theta \leftarrow \infty$;
		
		\For{$ V_{k,s} \in Q $ }{
			\If{$ D[k,s] < \theta $}{
				$ \theta \leftarrow D[k,s] $;
				
				$ k^* \leftarrow k $;
				
				$ s^* \leftarrow s $;
			}
		}
		
		$ Q \leftarrow Q \setminus \{V_{k^*,s^*}\} $;
		
		\For{$ V_{k^*,n} \in Q $}{
			\If{$ E_{k^*}(V_{k^*,n},V_{k^*,s^*})<\infty $}{
				$ \gamma \leftarrow D[k^*,s^*] + E_{k^*}(s^*,n)$;
				
				\If{$ \gamma < D(k^*,n) $}{
					$ D[k^*,n] \leftarrow \gamma $;
					
					$ P[k^*,n] \leftarrow V_{k^*,s^*} $;
				}
			}
		}
		
		\If{$ k^* +1 \leq |K| $}{
			$ \gamma \leftarrow D[k^*, s^*] + T_{k^*}(s^*)$;
			
			\If{$ \gamma < D[k^*+1,s^*] $}{
				$ D[k^*+1,s^*] \leftarrow \gamma $;
				
				$ P[k^*+1,s^*] \leftarrow V_{k^*,s^*} $;
			}
		}		
	}
	
	$ L_{\pi} \leftarrow D[|K|,v] $;
	
	$ V_{\hat{k},\hat{s}} \leftarrow P[|K|,v] $;
	
	\While{$ V_{\hat{k},\hat{s}} \neq V_{1,u} $}{
		$ \pi \leftarrow [V_{\hat{k},\hat{s}} , \pi] $;
		
		$ V_{\hat{k},\hat{s}} \leftarrow P[\hat{k},\hat{s}] $;
	}
\end{algorithm}

As Algorithm~\ref{ModifiedDijkstra} shows, the state graph $ G_{S} = (\overline{V}, \overline{E}, T ) $, the corresponding state set $ K $, the source node $ u $, and the destination node $ v $ are the input data.
Algorithm~\ref{ModifiedDijkstra} outputs the shortest path $ \pi $ from $ V_{1,u} $ to $ V_{|K|,v} $ and its length $ L_{\pi} $.

Algorithm~\ref{ModifiedDijkstra} contains three parts. 
The first part is the initialization (line 1 to line 7). $ Q $ is the unvisited node set. $ D $ is the set recording the tentative distance value (i.e., the length of the shortest path discovered so far) between $ V_{1,u} $ and each node in $ \overline{V} $. The tentative distance value is set to zero for our initial node and to infinity for all other nodes. In addition, $ P $ is the set of parent nodes of each node in $ \overline{V} $.

The second part is the shortest path searching (line 8 to line 34). First, we find the node $ V_{k^*,s^*} $ whose tentative distance value is the minimum among all unvisited nodes (Line 9 to line 16). Then, delete the above selected node from the unvisited nodes set $ Q $ (line 17). Next, we consider all the unvisited neighbors in state $ k^* $ of $ V_{k^*,s^*} $ and calculate their tentative distances through the selected node (line 18 to line 26). By comparing the newly calculated tentative distance to the one currently assigned to the neighbor, the smaller one is assigned to it. If the tentative distance of a neighbor is updated, its parent node needs to be changed to $ V_{k^*,s^*} $ accordingly.
Similarly, in line 27 to line 33, we consider the corresponding node of $ V_{k^*,s^*} $ in state $ k^*+1 $ and calculate the tentative distances of $ V_{k^*+1,s^*} $ through $ V_{k^*,s^*} $. The smaller of the newly calculated and the currently assigned tentative distance is assigned to $ D[k^*+1,s^*] $. If the tentative distance is updated, its parent node of $ V_{k^*+1,s^*} $ needs to be changed to $ V_{k^*,s^*} $ accordingly.

The third part is the output generation part (line 35 to line 40). Line 35 outputs the length of the shortest path between $ V_{1,u} $ and $ V_{|K|,v} $. From line 36 to line 40, by searching the parent node set iteratively, the shortest path between $ V_{1,u} $ and $ V_{|K|,v} $ is output.
\vspace{-0.5em}
\subsection{Algorithm Correctness}
In the following, we prove that the extended Dijkstra's algorithm is correct by induction. Let $ D[k,s] $ be the distance label found by the algorithm and let $ \delta[k,s] $ be the shortest path distance from $ V_{1,u} $ to $ V_{|K|,v} $. We want to show that $ D[k,s] = \delta[k,s] $ for every node $ s $ in any state $ k $ at the end of the algorithm, showing that the algorithm correctly computes the distances. We prove this by induction on $ |R|\ (R= \overline{V} \setminus Q) $ via the following lemma.

\begin{lemma}
	For each $ V_{k,s} \in R $, $ D[k,s] = \delta[k,s] $.
\end{lemma}

\begin{proof}
	\textit{Base case ($ |R| = 1 $):} Since $ R $ only grows in size, the only time $ |R| = 1 $ is when $ R = \{V_{1,u}\} $ and $ D[1,u] = 0 = \delta[1,u] $, which is correct.
	
	\textit{Inductive hypothesis:} Let $ V_{\epsilon,w} $ be the next node added to $ R $. Let $ R' = R \cup \{V_{\epsilon,w}\} $. Our I.H. is: for each $ V_{k,s} \in R' $, $ D[k,s] = \delta[k,s] $.
	
	\textit{Using the I.H.:} By the inductive hypothesis, for every node in $ R' $ that isn't $ V_{\epsilon,w} $, the distance label is correct. We need only show that $ D[\epsilon,w] = \delta[\epsilon,w] $ to complete the proof.
	
	Suppose a contradiction that the shortest path from $ V_{1,u} $ to $ V_{\epsilon,w} $ is $ O $ and has length $ L(O) < D[\epsilon,w] $. $ O $ starts in $ R' $ and at some leaves $ R' $ (to get to $ V_{\epsilon,w} $ which is not in $ R' $). Let $ e_{x,y}=(V_{k_x,s_x},V_{k_y,s_y})\ (k_y-k_x \in \{0,1\})$ be the first edge along $ O $ that leaves $ R' $. Let $ O_x $ be the sub-path between $ V_{1,u} $ and $ V_{k_x,s_x} $ of $ O $. Clearly: $L(O_x) + L(e_{x,y}) \leq L(O) $. (If $ k_x = k_y$, $ L(e_{x,y}) = E_{k_x}(s_x,s_y)$; otherwise, $ L(e_{x,y}) = T_{k_x}(s_x) $.)
	Since $ D[k_x,s_x] $ is the length of the shortest path from $ V_{1,u} $ to $ V_{k_x,s_x} $ by the I.H., $ D[k_x,s_x] \leq L(O_x) $, giving us $ D[k_x,s_x] + L(e_{x,y}) \leq L(O_x) + L(e_{x,y}) $. Since $ V_{k_y,s_y} $ is adjacent to $ V_{k_x,s_x} $, $ D[k_y,s_y] $ must have been updated by the algorithm, so $ D[k_y,s_y] \leq D[k_x,s_x] + L(e_{x,y}) $. Finally, since $ V_{\epsilon,w} $ was picked to be the next node added to $ R $ by the algorithm, $ V_{\epsilon,w} $ must have the smallest distance label: $ D[\epsilon,w] \leq D[k_y,s_y] $. Combining these inequalities in reverse order, a contradiction $ D[k_x,s_x] < D[k_x,s_x] $ is formed. Thus, no such shorter path $ O $ must exist. So $ D[k,s] = \delta[k,s] $. 	$\hfill\blacksquare$
\end{proof}

This lemma shows the algorithm is correct by applying the lemma for $ R = \overline{V} $.
\vspace{-0.5em}
\subsection{Algorithm Complexity}
The computational complexity of the proposed extended Dijkstra's algorithm on state graph $ G_{S} = (\overline{V}, \overline{E}, T) $ mainly resides in lines 8--34. 
Line 8 indicates that lines 9--33 will be executed $ |Q| $ times. 
Since there are $ \Theta(|\overline{V}|) $ nodes in set $ Q $, the complexity of executing lines 10--16 $ |Q| $ times is $ \Theta(|\overline{V}|^2) $. 
Lines 18--19 restrict $ E_{k^*}(V_{k^*,n},V_{k^*,s^*}) $ to the edges connecting $ V_{k^*,s^*} $ and its adjacent nodes in the same state; thus, the complexity of executing lines 18--26 $ |Q| $ times is $ \Theta(|K|\cdot|E|) $. 
Since there's no loop in lines 27--33, the complexity of executing lines 27--33 $ |Q| $ times is $ \Theta(|\overline{V}|) $. 
According to the definition of $ G_{S} $, $ |\overline{V}| = |V| \cdot |K| $ and $ |E| = |V|^2 $. Therefore, the computation complexity of the proposed algorithm is $ \Theta(|K|^2\cdot|V|^2) $. 

\section{Adaptive Offloading With State Graph}\label{AdaptiveScheduling}
In this section, we utilize the mathematical tools proposed in Section~\ref{StateGraph} to solve the space task offloading problem.
As the optimal offloading targets can be different for tasks with diverse transmission and computational requirements, we propose an adaptive offloading scheme based on state graph. 

\begin{definition}[Task]
	Any task $ \tau $ has five attributes, whose set is $ \Lambda_ {\tau} = (u,v,t,\widetilde{C},\widetilde{N},\widetilde{N'}) $. $ u $ and $ v $ are the source and destination, respectively. $ t $ is the instance when $ \tau $ is generated. $ \widetilde{C} $ is the computational requirement of $ \tau $. $ \widetilde{N} $ and $ \widetilde{N'} $ are data volumes of $ \tau $ before and after computation, respectively. Tasks are the smallest unit of transmission and computation.
\end{definition}

As stated in Section~\ref{Introduction}, computational states need to be distinguished to select the optimal offloading target. Therefore, we construct the network shown in Fig.~\ref{NetworkModel} as a state graph with two states. By defining the two states as the \emph{uncomputed} state and \emph{computed} state, the task computation process can be expressed as a migration between the above two states. 
Next, we introduce the state graph $ G_{S} = (\overline{V}, \overline{E}, T) $ constructed based on the network shown in Fig~\ref{NetworkModel} in detail.

$ \overline{V} = \{ V_{k,s}|k \in \{1,2\},s\in V\cup u \cup v \} $. $ 1 $ and $ 2 $ represent the uncomputed and computed states, respectively. $ V $ is the set of LEO satellites, $ u,v $ are the source and destination\footnote{Since the data is usually transmitted between the ground station, the ground server and the institution over fiber, the transmission delay among them is negligible. In this condition, the ground station, the ground server, and the institution are simplified as the destination node in the state graph.} of task $ \tau $, respectively. 
	
$ \overline{E} = \{ E_{k}\ |\ k\in K\}$. $  E_{k}(s,s') ( s,s'\in V\cup u \cup v)$ satisfies
	\begin{small}
			\setlength{\abovedisplayskip}{3pt}
			\setlength{\belowdisplayskip}{3pt}
			\begin{equation}
			\left\{
			\begin{array}{ll}
				\widetilde{N}=\int_{t_s}^{t_s+E_k} R(s,s',t) \mathrm{d}t,\ \mathit{k}=1, \\
				\widetilde{N'}=\int_{t_s}^{t_s+E_k} R(s,s',t) \mathrm{d}t,\ \mathit{k}=2,
			\end{array}\label{Equ:2}
			\right.
		\end{equation}
	\end{small}%
	where $ t_s $ is the instant that task $ \tau $ arrives at node $ s $. 
	$ R(s,s',t) $ is the available data transmission rate of edge $ e_{s,s'} $ at instant $ t $, which is independent from states. Depending on the node types of $ s $ and $ s' $, $ e(s,s') $ can be an ISL or satellite-to-ground link (SGL). 
	If $ s $ and $ s' $ are not communicable\footnote{Two nodes are communicable or not is determined by the distance between them and the communication range of the signal transceiver.} (i.e., $ s \nLeftrightarrow s' $), $ R(s,s',t) =0 $; otherwise, $ R(s,s',t) \in [0,R_{max}] $, where $ R_{max} $ is the maximum data transmission rate of $ e_{s,s'} $.
	
$ T = \{ T_{k}\ |\ k\in K \} $. $ T_k(s) $ satisfies
	\begin{small}
		\setlength{\abovedisplayskip}{3pt}
		\setlength{\belowdisplayskip}{3pt}	
		\begin{equation}
			\widetilde{C}=\int_{t_s}^{t_s+T_{k}}C(s,t)\mathrm{d}t,
		\end{equation}
	\end{small}%
	where $ C(s,t) $ is the available computing capability of node $ s $ at instant $ t $. Since the source node does not perform computation, when $ s=u $, $ C(s,t) = 0 $. Additionally, ground servers (such as supercomputers) usually have large amounts of computing resources so that they can complete the computation in a very short time; therefore, when $ s=v $, $ C(s,t) = \infty $. In other cases, $ C(s,t) \in [0,C_{s}^{max}] $, where $ C_{max} $ is the maximum computing capability of each satellite. 
	
In other words, $ E_{k}(s,s') $ represents the edge weight between different nodes in the same state, whose value equals the delay of transmitting task $ \tau $ from node $ s $ to $ s' $; whereas, $ T_k(s) $ represents the weight of edge from a node in the uncomputed state to itself in the computed state, whose value equals the delay of computing task $ \tau $ on node $ s $. In this way, both the transmission and computation processes can be represented in the state graph. It should be noted that, we fit the continuous integration to our discrete model with discretization, where minor error may be introduced.

In the state graph above, the task offloading problem is transformed into a shortest path problem. 
By performing the extended Dijkstra's algorithm on the constructed stated graph, the transmission and computation delays are jointly optimized. In this way, the shortest path, which contains the optimal offloading target, is obtained. Since the constructed state graph contains both conventional and extended offloading targets, the optimal offloading decision can be made adaptively according to the characteristics of tasks.


\section{Performance Evaluation}\label{evaluation}
In this section, we evaluate the performance of the proposed \emph{adaptive offloading scheme} and compare it with the \emph{ground} and \emph{one-hop} offloading scheme. 

In this section, the Walker constellation has 8 orbits, each with 16 satellites, located at an altitude of 500 km. Based on the state-of-the-art techniques, $ R_{max}^{ISL} $ and $ R_{max}^{SGL} $ are set to 5 Gbps~\cite{del2019technical} and 1 Gbps~\cite{yost2021state}, respectively. 
%
Unless otherwise stated, the onboard computing platform is SpaceCube v2.0~\cite{SpaceCube2.0} whose maximum computing capability is 200 GFLOPS. 
The task arrival of each area in space is a Poisson stochastic process whose arrival rate is estimated based on the network connection data from Internet Census~\cite{InternetConsensus2012}. The sum of arrival rates in space is set to 1000. The destination locations of tasks are randomly generated. $ \widetilde{N'} $ is set to 16 bits.

\begin{figure}[h]
	\centering
	\vspace{-0.5em}
	\setlength{\abovecaptionskip}{-0.cm}
	\setlength{\belowcaptionskip}{-0.cm}
	\includegraphics[width=0.50\textwidth]{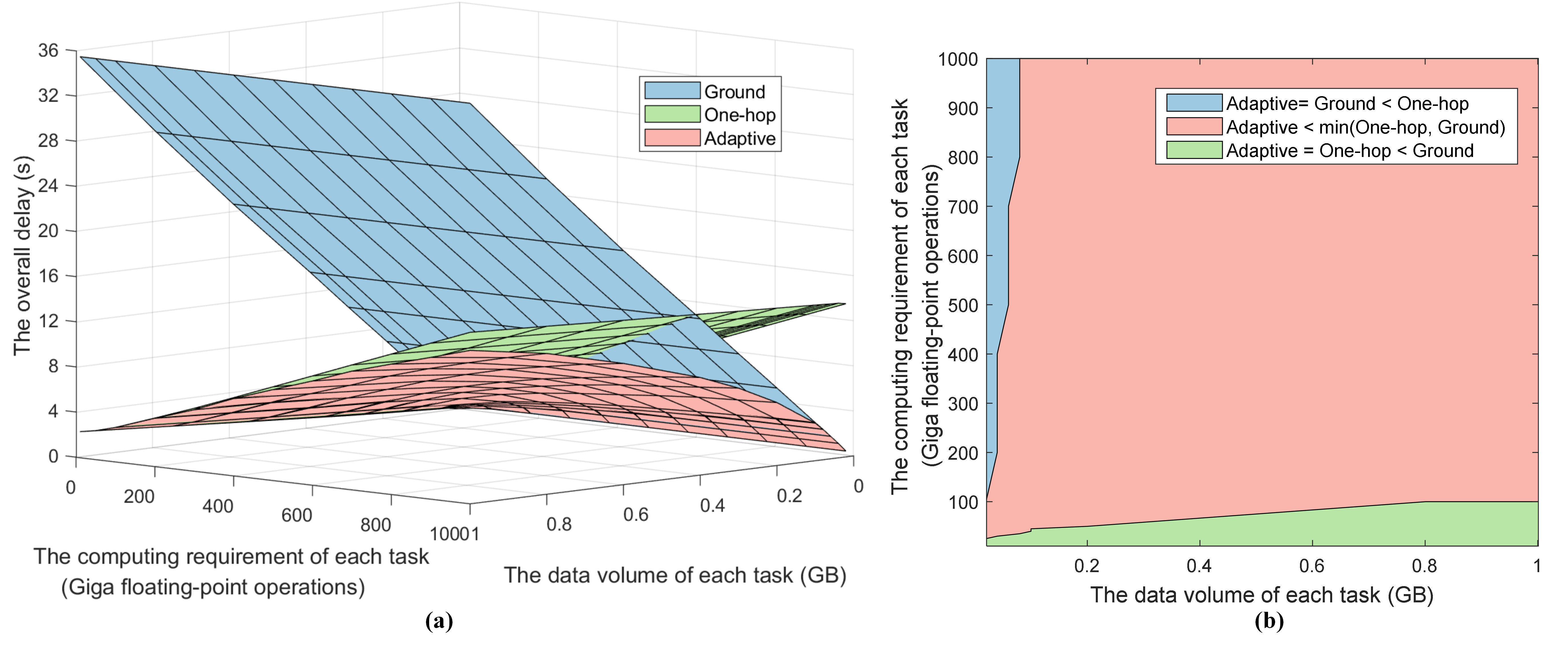}\\ 
	\begin{center}
		\caption{(a) The overall delay versus the data volume $ \widetilde{N} $ and computational requirement $ \widetilde{C} $ of tasks. (b) The offloading scheme which achieves the lowest overall delay.}\label{SimFig1}
	\end{center}
	\vspace{-1.5em}	 
\end{figure}
Fig.~\ref{SimFig1} (a) shows the overall delay versus the data volume $ \widetilde{N} $ and computational requirement $ \widetilde{C} $ of tasks. 
Since the red surface lies below the blue and green surfaces, it can be concluded that \textbf{the overall delay of the proposed scheme is always less than or equal to the benchmark schemes}. 
Especially, when $ \widetilde{N} = 0.4 $ GB and $ \widetilde{C}=1000 $ Giga floating-point operations (GFLO), the proposed adaptive offloading scheme decreases the overall delay by 37.56\% and 39.35\% compared with the ground offloading scheme and the one-hop offloading scheme, respectively.

\textbf{From the dimension of data volume, the proposed scheme degrades to the ground offloading scheme when $ \widetilde{N} $ is small and to the one-hop offloading scheme when $ \widetilde{N} $ is large.}
When $ \widetilde{N} $ is small, the advantage of computing with powerful ground servers outweighs the disadvantage of transmitting raw data across the network. When $ \widetilde{N} $ is large, the cost of transmitting raw data increases dramatically, the advantage of reducing transmission delay by offloading tasks to satellites within one-hop range outweighs the disadvantage that these satellites can only provide limited computing capability.
\textbf{From the dimension of computational requirement, the proposed scheme degrades to the one-hop offloading scheme when $ \widetilde{C} $ is small.} In this situation, satellites within one-hop range are already able to provide sufficient computing resources to ensure that the saved transmission delay outweighs the increased computation delay.
Except for these extreme cases, the proposed scheme outperforms the benchmark schemes, whose reasons will be introduced in the analyses of Fig.~\ref{SimFig2}.

Fig.~\ref{SimFig1} (b) shows the offloading scheme which achieves the lowest overall delay for different tasks. It can be concluded that \textbf{the proposed adaptive offloading scheme outperforms the benchmark schemes over a wide range of parameters} (i.e., the red area). In the blue and green area, the proposed degrades to the ground offloading scheme and the one-hop offloading scheme, respectively. 

\begin{figure}[h]
	\centering
	\vspace{-0.5em}
	\setlength{\abovecaptionskip}{-0.cm}
	\setlength{\belowcaptionskip}{-0.cm}
	\includegraphics[width=0.48\textwidth]{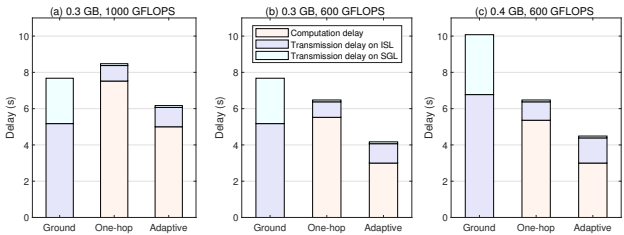}\\ 
	\begin{center}
		\caption{Breakdown of the overall delay.}\label{SimFig2}
	\end{center}
	\vspace{-1.5em}	 
\end{figure}



Fig.~\ref{SimFig2} shows the breakdown of the overall delay of the proposed and benchmark offloading schemes. 
Here the transmission/computation delay includes the waiting delay between the arrival instant of the corresponding node and the start of transmission/computation, respectively. 

\textbf{The proposed scheme outperforms the ground offloading scheme because the data volume is greatly reduced after the onboard computation}. As shown in Fig.~\ref{SimFig2}, the transmission delays on ISL and SGL of the proposed scheme is much lower than the ground offloading scheme. In the meanwhile, the increase in computation delay is less than the decrease in transmission delay. 
\textbf{The proposed scheme outperforms the one-hop offloading scheme because the proposed scheme can utilize computing resources in a larger area to relieve local overloads}. As shown in Fig.~\ref{SimFig2}, the computation delay of the adaptive offloading scheme is much lower than that of the one-hop offloading scheme. In the meanwhile, the increase in transmission delay on ISL is less than the decrease in computation delay. 

\begin{table}[h]
	\centering
	\vspace{0.5em}
	\begin{spacing}{1}
		\captionsetup{font={footnotesize},justification=centering}
		\caption{Performance Improvements on Existing Highly Integrated Onboard Computing Platforms. ($ \widetilde{N} =  0.3$ GB, $ \widetilde{C} = 1000 $ GFLO)}\label{Table2}
	\end{spacing}
	\small
	\resizebox{0.45\textwidth}{!}{

\begin{tabular}{|c|c|c|c|} 
	\hline
	\multirow{3}{*}{Onboard Computing Platform} & \multirow{3}{*}{\begin{tabular}[c]{@{}c@{}}Computing\\Capability\\(GFLOPS)\\\end{tabular}} &  \multicolumn{2}{c|}{\multirow{2}{*}{Improvements}}\\ 
& & \multicolumn{2}{c|}{} \\
	\cline{3-4}
	&                                                                                          & Ground & One-hop                   \\ 
	\hline
	Unibap iX5-100                              & 127                                                                                      & +21\%  & +159\%                    \\ 
	\hline
	SpaceCube v2.0                              & 200                                                                                      & +42\%  & +93\%                     \\ 
	\hline
	SpaceCube v3.0                              & 590                                                                                      & +139\% & +15\%                     \\ 
	\hline
	MOOG V-Series Ryzen                         & 1000                                                                                     & +256\% & +6\%                      \\ 
	\hline
\end{tabular}}
	\vspace{1em}
\end{table}

Table~\ref{Table2} evaluates the performance of the adaptive offloading scheme over the benchmark schemes on several existing onboard computing platforms. (+0\% means a tie, whereas +100\% means saving half the time.)
\textbf{As the computing capability increases, the performance improvement over the ground offloading scheme increases and the improvement over the one-hop offloading scheme decreases.}
This is because the computation delay of LEO satellites decreases when their computing capability increases, making the disadvantage of limited onboard computing resources less severe. In the meanwhile, the increase in computing capability leads to a reduction in local overloads; thus, fewer tasks need to be offloaded to satellites beyond one-hop range for computation.

	

\section{Conclusion}\label{Conclusions}
We have proposed an adaptive offloading scheme based on state graphs for space missions. The proposed scheme outperforms the benchmark schemes for being able to escape the transmission-computation cost dilemma by extending offloading targets with LEO satellites beyond one-hop range and implementing joint optimization of transmission and computation. Its low computational complexity offers the possibility of application in real networks.
\bibliographystyle{IEEEtran}
\bibliography{IEEEfull,Reference}

\end{document}